\documentclass[a4paper,10pt,reqno]{amsart}
\usepackage[utf8]{inputenc}
\usepackage[%
    hyperref,backend=biber, doi=false,url=false,
    sorting=none,style=numeric-comp,backref=true, maxbibnames=7,giveninits=true]{biblatex} 
\renewbibmacro{in:}{}

\usepackage{amsmath}
\usepackage{amssymb}
\usepackage{bbm}
\usepackage{bm}
\usepackage{color}
\usepackage{chngcntr}
\numberwithin{equation}{section}
\usepackage[foot]{amsaddr}

\usepackage{mathrsfs}
\usepackage{graphicx}
\usepackage{epstopdf}
\usepackage{hyperref}
\hypersetup{colorlinks=true,linktoc=page,linkcolor=blue,citecolor=red,urlcolor=cyan}

\newcommand{\old}[1]{ \ifthenelse{\equal{}{veruno}}{\textcolor{magenta}{#1}}{#1} }
\newcommand{\colred}[1]{\textcolor{red}{(\sc #1)}}

\bibliography{bibliografia}

\usepackage[english]{babel}

\newcommand{\hu}{\hspace{1cm}}

\newcommand{\perpe}[1][x]{{#1}_d}
\newcommand{\dime}{D}
\newcommand{\dk}[1][k]{\frac{{\rm d} #1}{(2\pi)}}
\newcommand{\dx}[1][x]{{{\rm d} #1}}

\newcommand{\mathi}{{\operatorname{i}}}
\newcommand{\Tr}{\operatorname{Tr}}

\newcommand{\der}{D}
\newcommand{\genlaplace}{\mathcal{D}_0}
\newcommand{\hkf}{H}
\newcommand{\ope}{\mathcal{D}}
\newcommand{\Kope}{K_{\ampli,V}}

\newcommand{\ampli}{\lambda}

\newcommand{\Kdelta}[1][]{K_{\delta}}
\newcommand{\Kzeta}[1][\ampli]{K_{#1}}

\newcommand{\pos}[1][]{L_{#1}}

\newcommand{\hidepdo}[2]{\ifthenelse{\equal{jpa}{}}{#1}{#2}}

\usepackage{xparse}

\NewDocumentCommand{\dkpara}{O{k} O{\dime}}{\frac{{\rm d} {#1}^{\parallel}}{(2\pi)^{#2}}}
\NewDocumentCommand{\dkd}{O{k} O{\dime}}{ \frac{{\rm d}^{#2} {#1}}{(2\pi)^{#2}}}
\NewDocumentCommand{\dxd}{O{x} O{\dime}}{ {\rm d}^{#2} {#1} }
\NewDocumentCommand{\ad}{O{} O{}}{ \bigl \langle {#1}\bigr\rangle_{{\rm ad}{#2}} }

\makeatletter
\def\@tocline#1#2#3#4#5#6#7{\relax
  \ifnum #1>\c@tocdepth 
  \else
    \par \addpenalty\@secpenalty\addvspace{#2}%
    \begingroup \hyphenpenalty\@M
    \@ifempty{#4}{%
      \@tempdima\csname r@tocindent\number#1\endcsname\relax
    }{%
      \@tempdima#4\relax
    }%
    \parindent\z@ \leftskip#3\relax \advance\leftskip\@tempdima\relax
    \rightskip\@pnumwidth plus4em \parfillskip-\@pnumwidth
    #5\leavevmode\hskip-\@tempdima
      \ifcase #1
       \or\or \hskip 2em \or \hskip 4em \else \hskip 6em \fi%
      #6\nobreak\relax
    \dotfill\hbox to\@pnumwidth{\@tocpagenum{#7}}\par
    \nobreak
    \endgroup
  \fi}
\makeatother

\title[Resummed heat-kernel and form factors for surface contributions]{Resummed heat-kernel and form factors for surface contributions: Dirichlet semitransparent boundary conditions}

\author{S.~A.~Franchino-Vi\~nas$^{1,2,3}$}

\address{$^1$ Helmholtz-Zentrum Dresden-Rossendorf, Bautzner Landstraße 400, 01328 Dresden, Germany.}

\address{$^2$ Departamento de F\'isica, Facultad de Ciencias Exactas,
Universidad Nacional de La Plata, C.C.\ 67 (1900), La Plata, Argentina.}

\email{$^3$ \href{mailto:s.franchino-vinas@hzdr.de}{s.franchino-vinas@hzdr.de} }

\newtheorem{comm}{Remark}
\newtheorem{theorem}{Theorem}
\newtheorem{proposition}{Proposition}
\newtheorem{corollary}{Corollary}

\newtheorem{definition}{Definition}
\begin{document}

\begin{abstract}
In this article we consider resummed expressions for the heat-kernel's 
trace of a Laplace operator, the latter including a potential and imposing Dirichlet semitransparent boundary conditions on a surface of codimension one in flat space. 
We obtain resummed expressions that correspond to the first and second order expansion of the heat-kernel in powers of the potential.
We show how to apply these results to obtain the bulk and surface form factors  of a scalar quantum field theory in $d=4$ with a Yukawa coupling to a background. \hidepdo{}{A characterization of the form factors in terms of pseudo-differential operators is given.}\old{Additionally, we discuss a connection between heat-kernels for Dirichlet semitransparent, Dirichlet and Robin boundary conditions. }
\end{abstract}

\maketitle

\section{Introduction}\label{sec:introduction}
The relation between spectral functions and the quantum theory of fields has been a\old{close}one~\cite{DeWitt:1965,Birrel:1982}, 
specially when considering external background fields, including electromagnetic fields~\cite{Schwinger:1951nm} and curved spacetimes~\cite{Dowker:1975tf,Hawking:1976ja, Gibbons:1977dr}.
In particular, the (one-loop) quantum fluctuations are usually given in terms of operators of Laplace type,
whose heat-kernels determine the one-loop effective action~\cite{Peskin:1995ev}. 

Recall that, given a  Laplace-type operator $\genlaplace$ defined on a real manifold $M$ of dimension $m$, with or without boundary and the corresponding local boundary conditions, 
its heat-kernel (HK) is defined as $K(T;\genlaplace):=e^{-T \genlaplace}$.
The theory establishes that under general conditions of smoothness of both the operator and the manifold, 
as $T\downarrow 0$ the trace of the HK
possesses the following expansion~\cite{Greiner:1971xe,Gilkey:1975iq,Gilkey:1995,Kirsten:2001wz}:
\begin{align}\label{eq:hk_expansion}
 K(T;f;\genlaplace):= \Tr \Big(f(x) K(x,y; T;\genlaplace)\Big)= \sum_{n=0} a_n(f;\genlaplace) T^{(n-m)/2}.
\end{align}
The function $f$ is called smearing function; 
its role is to give a precise mathematical meaning to terms that otherwise would be divergent (in physical terms, it acts as a regulator).
The coefficients $a_n(f;\genlaplace)$ are called Gilkey--Seeley--DeWitt (GSDW) coefficients~\cite{DeWitt:1965,Seeley:1967ea,Gilkey:1975iq}
and sometimes HAMIDEW, after Hadamard--Minakshisundaram--DeWitt~\cite{Gibbons:1979}.
They consist of volume and surface integrals, respectively over the bulk and the boundary of $M$, of local invariants (including the smearing function).
One can build them by considering a linear combination of all the possible invariants with the appropriate dimensions; 
the numerical coefficients in front of each single term are universal, 
i.e. independent of the problem at hand.
Additionally, it can be shown that 
the HK corresponds to the solution of the following heat equation with initial localized conditions:
\begin{align}\label{eq:heat-equation}
 \begin{split}
  (\partial_T+ \genlaplace)K(x,y;T;\genlaplace)=0,\quad K(x,y;\downarrow 0;\genlaplace)= \delta(x,y).
 \end{split}
\end{align}

In the last decades, the techniques in the computation of the GSDW coefficients have shown several advancements. 
During this period, the community has recognized the benefits of the joint use of index theorems and functorial techniques, 
together with the consideration of special cases~\cite{Vassilevich:2003xt}.

Another important milestone 
has been the\old{obtainment}of partially resummed HK expansions.
For example, it has been proved that, in the so-called covariant perturbation theory, one can resum all the derivatives acting on contributions to second~\cite{Barvinsky:1987uw,Barvinsky:1990up,Codello:2012kq} and third order in the curvatures~\cite{Barvinsky:1990uq,Barvinsky:1993en,Barvinsky:1994ic,Barvinsky:1994hw} (see also rederivations and physical consequences in~\cite{Avramidi:1990ap,Avramidi:1990je}).
Other studied scenarios include resummations in abelian bundles~\cite{Avramidi:2009quh}, \old{QED~\cite{Gusynin:1998bt}}, symmetric spaces \cite{Avramidi:1994zp, Avramidi:1995ik} and powers of the curvature~\cite{Parker:1984dj,Jack:1985mw}; see Ref.~\cite{Avramidi:1997jy} for additional considerations\old{and~\cite{Wigner:1932eb,Fujiwara:1981rf} for a related computation by Wigner.}

In the above-mentioned developments, much more attention has been given to the case of manifolds without boundaries, 
being the study of HKs in manifolds with boundaries much less developed; 
a not exhaustive list of works which deal with the latter problem include Refs.~\cite{Vassilevich:2003xt,Kirsten:2004qf,Esposito:2004ts,Esposito:2005dn,McAvity:1990we,McAvity:1991xf,McAvity:1992fq, Barvinsky:2005ms, Barvinsky:2009cb} and references therein.

Trying to bring more balance into this scenario, in the present manuscript we will show how to obtain resummed expressions when Dirichlet semitransparent conditions on a flat surface are considered.
According to the best of our knowledge, 
this is the first time that resummations for surface contributions of bulk quantities are studied.

The Dirichlet semitransparent boundary condition, also known as transmittal boundary condition \cite{Gilkey:2001mj}, is equivalent to the introduction of a delta potential with support on a surface of codimension one~\cite{Franchino-Vinas:2020okl, Asorey:2004kk, Albeverio:1988}. 
The problem at hand has been widely studied in the realm of quantum field theories (QFTs); the interested reader may consult Refs.~\cite{Fosco:2019lmw, Fosco:2017jjf, Parashar:2017sgo, Fosco:2011rm, Fosco:2007nz, Milton:2007gy, Bordag:2004rx, Graham:2002xq, Moss:2000gv, Frolov:1999bi, Bordag:1998vs, Solodukhin:1997xn}. 
Moreover, it has also been considered as a problem in a 
first quantization~\cite{Grosche:1990um, Lawande:1988, Manoukian:1989, Farias:1980, Reading:1972, Crandall:1993, Goovaerts:1985,Cacciari:2006, Cacciari:2007}. 
Recently, it has attracted much attention in connection with the related $\delta'$ problem, 
see Refs.~\cite{Cavero-Pelaez:2020pxy, Munoz-Castaneda:2020wif, Bordag:2019vrw} for $\delta-\delta'$ potentials 
and Ref.~\cite{Ahmadiniaz:2022bwy} for some generalizations of semitransparent boundary conditions.

To be  precise, in the following we will thus be interested in an operator of Laplace type in $d$-dimensional flat Euclidean space $\mathbb{R}^d$, 
whose potential is given by a sufficiently smooth function $V$, summed to a Dirac delta function with support in a\old{flat}surface of codimension one (chosen without loss of generality as the $\perpe[x]=L$ plane):
\begin{align}\label{eq:operator_D}
 \ope :&= -\partial^2+ \lambda \delta(\perpe[x]-\pos) - \zeta V(x).
\end{align}
Regarding its HK and physical applications, we may summarize the main results of the present article as follows:
\begin{itemize}
 \item Proposition~\ref{prop:integral_equation}, in Sec.~\ref{sec:preliminary},
 provides an integral equation for the relevant HK, from which an expansion
 in powers of the potential $V$ is derived in Corollary~\ref{cor:perturbative}.
 \item Theorem~\ref{th:firs_order}, in Sec.~\ref{sec:first_order}, gives a closed expression for the resummed-in-$\lambda$ HK at first order in $V$ and in $d=1$, 
 together with the corresponding GSDW coefficients.
 \item Theorem~\ref{th:quadratic}, Sec.~\ref{sec:second_order}, concerns the second order in $V$ and resummed-in-$\lambda$ contribution to the HK in $d=1$; 
the corresponding GSDW coefficients are also listed.
 \item In Sec.~\ref{sec:qft}, we apply the previous results to a scalar quantum field theory in $d$ dimensions, 
 including a Yukawa coupling to a background field. We show how the resummed expansions may be applied to obtain the relevant form 
 \hidepdo{factors.}{factors, 
 which are shown to be  pseudo-differential operators with symbols in $S^j$ for all $j>0$ (in the classification of 
 L.~Hörmander~\cite{Hoermander:LPDOIII}).}
\end{itemize}

\old{ In order to derive these results, we will require the potential in Eq.~\eqref{eq:operator_D} to be positive and $V$ to decay rapidly enough at infinity.
On the one hand, positivity can be guaranteed if one chooses a coupling $\lambda\geq0$ and enforces $\zeta V(x)<~0$; 
otherwise, negative eigenvalues would threaten the definition of the Euclidean quantum field theory in Sec.~\ref{sec:qft} by creating instabilities.
This requirement can be left aside in Secs.~\ref{sec:preliminary}, \ref{sec:first_order} and \ref{sec:second_order},
where we are just interested in the expression of the heat-kernel on the diagonal.
On the other hand, the decay of $V$ should be understood as enabling an interpretation of the potential as a perturbation of the Laplacian-plus-delta operator.
}

\section{A perturbative  expansion of the heat-kernel in powers of $V$}\label{sec:preliminary}

In the following sections, unless otherwise stated, we will refer to the case $d=1$.
We will comment on the possibility of applying our results to $d$ dimensional operators 
and come back to an arbitrary dimension in Sec.~\ref{sec:qft}, when we will study a QFT.
To fix the notation, we will call $\Sigma$ the support of the delta function; 
by analogy with the $d$-dimensional case, for $d=1$ we will use the notation
\begin{align}
 \int_\Sigma \dxd[x][0]\,g(x) :=\left\{g(x)\right\}[\Sigma]:= g(L).
\end{align}\old{
Additionally, to enhance the readability of expressions we will make two definitions; 
one takes advantage of the symmetry under exchange of two variables and the other regards the derivative operator:
\begin{align}
 g(x,y)+ g(y,x) &=: \left[ g(x,y)+ \{x\leftrightarrow y\}\right],
 \\
\der_x:&=-\mathi \partial_x.
\end{align}
}

As stated in the introduction, we will show how to compute the HK of the operator $\ope$ to all order in the $\ampli$ parameter, 
albeit performing an expansion in powers of $\zeta$.
To zeroth order in $\zeta$,
the HK of the operator has already been computed in closed form in the literature, see Refs.~\cite{Bauch:1985, Gaveau:1986, Franchino-Vinas:2020okl}.
Since it will prove crucial in our computations, we state this result.

\begin{proposition}\label{prop:hk_delta}
The HK for the Laplace-type operator $\ope_{\zeta=0}$ is given by
\begin{align}
 \begin{split}\label{eq:HK_delta}
&\Kzeta(x,y;T)
 = K_0(x,y;T)- \frac{\ampli}{2}\int_{\mathbb{R}^+} \dx[u]\, e^{-\frac{u \ampli}{2}} K_0(|x-L|+|y-L|+u,0;T),
 \end{split}
\end{align}
where $K_0(x,y;T):=(4\pi T)^{-1/2}{\operatorname{exp}\left({-\frac{(x-y)^2}{4T}}\right)}$ corresponds to the free HK in  one-di\-men\-sio\-nal flat space.
\end{proposition}

\begin{corollary}\label{cor:trace_hk_delta}
As an immediate consequence of Proposition~\ref{prop:hk_delta}, one can prove that the functional trace of $\ope_{\zeta=0}$'s  HK  is given by (see~\cite{Munoz-Castaneda:2013yga})
\begin{align}
\label{eq:HK_delta_trace}
&\Tr \Big( \Kzeta(x,y;T)-K_0(x,y;T)\Big)
 = \frac{1}{2} \int_{\Sigma}\dxd[x][0]\left[e^{\frac{\lambda ^2 T}{4}} \operatorname{erfc}\left(\frac{\lambda  \sqrt{T}}{2}\right)-1\right].
\end{align} 
\end{corollary}

\begin{theorem}\label{cor:trace_hk_delta_smearing}
 \old{
 If one adds a smearing function, then the  trace of the operator $\ope_{\zeta=0}$'s HK, under the assumptions around Eq.~\eqref{eq:operator_D}, is given by 
 \begin{align}
 \begin{split}\label{eq:trace_hk_zeroth}
&\Tr\Big( f(x)\Kzeta(x,y;T)  \Big)
 =
 \frac{1}{\sqrt{4\pi T}}\int_{\mathbb{R}} \dx[x] f(x) + \int_{\Sigma} \dxd[x][0]  
 \hkf_{0,\Sigma}^{(f)}(\der_x;T;\lambda) f(x),
 \end{split}
 \end{align}
where we have defined the kernel
\begin{align}\label{eq:kernel_order0}
  \hkf_{0,\Sigma}^{(f)}(k;T;\lambda):&=\frac{\lambda e^{-\frac{k^2 T}{4} }}{2 \left(k^2+\lambda ^2\right)} \Bigg[\lambda e^{\frac{(\lambda ^2 +k^2 )T}{4} } \operatorname{erfc}\left(\frac{\lambda  \sqrt{T}}{2}\right)-{k  \operatorname{erfi}\left(\frac{k \sqrt{T}}{2}\right)}-{\lambda }\Bigg].
 \end{align}
}
\end{theorem}

\begin{corollary}\label{cor:order0_SDW_coef}
 \old{ 
 The only nonvanishing bulk GSDW coefficient associated to the operator $\ope_{\zeta=0}$, can be readily read from Eq.~\eqref{eq:trace_hk_zeroth}.
 The surface contributions to the GSDW coefficients  
vanish for $n<1$.
For $n\geq 1$ an odd integer, they read
\begin{align}\label{eq:order0_SDW}
 \begin{split}
&a_{n,\Sigma}^{(0)}(f;\ope_{\zeta=0})
=
\Bigg\lbrace
 \frac{\lambda  }{2^{n} \Gamma \left(\frac{n+1}{2}\right)} \frac{  \left( \lambda ^n-\lambda  \partial^{n-1}\right)}{ \left(\lambda ^2-\partial^2\right) } f(x)
 \Bigg\rbrace[\Sigma]
 ,
 \end{split}
 \end{align}
 while, for $n\geq 2$ and even, we have
 \begin{align}\label{eq:order0_SDW2}
 \begin{split}
& a_{n,\Sigma}^{(0)}(f;\ope_{\zeta=0})
 =\Bigg\lbrace
 \frac{\lambda  }{2^{n}\Gamma \left(\frac{n+1}{2}\right)}\frac{  \big(\partial^n- (-\lambda )^n\big)}{\left(\lambda ^2-\partial^2\right) } f(x)
 \Bigg\rbrace[\Sigma].
\end{split}
\end{align}
}
\end{corollary}
\old{ 
In order to check the validity of these results, 
one can compare the first coefficients with those previously computed in Ref.~\cite{Gilkey:2001mj}.
Notice also that the presence in the coefficients  of a quotient involving derivatives is only apparent: 
considering $\partial$ as a symbol, one should first perform the division; the latter is exact for every $n$,
showing thus that the coefficients are made of local terms.
}

{
\begin{comm}\label{comm:order0}
 Taking into account the universality of the numerical coefficients in the HK expansion, 
 one can lift the one-dimensional results in Corollary~\ref{cor:order0_SDW_coef} to arbitrary dimension, 
 including the possibility of introducing additional geometrical features such as a curved manifold $M$ and a curved surface $\Sigma$. 
 Then: 
 \begin{enumerate}
  \item The dependence on the dimension would be only through an overall factor $(4\pi)^{-m/2}$, see Ref.~\cite{Bordag:1999ed}.
\item  The coefficients in Eqs.~\eqref{eq:order0_SDW} and \eqref{eq:order0_SDW2} would correspond to integrals over a $(d-1)$-dimensional hypersurface $\Sigma$ with the appropriate measure.
The variable $\lambda$ may be taken as a space-dependent quantity (living on the surface $\Sigma$) and the derivatives in \eqref{eq:order0_SDW} and \eqref{eq:order0_SDW2} become covariant derivatives normal to $\Sigma$.
 The integrals over $\mathbb{R}$ would be replaced by integrals over $M$ with the corresponding measure.
\item  The additional geometrical features would manifest themselves as the appearance of new nonvanishing invariants contributing to the GSDW coefficients, including curvatures of the spacetime, 
 the extrinsic curvature (also called second fundamental form) of the hypersurface $\Sigma$, covariant derivatives in the directions tangent to $\Sigma$ acting on $\lambda$, etc.
 \end{enumerate}
\end{comm}
}

The situation becomes more involved if one tries to perform the computation to higher orders in $\zeta$.
One powerful technique that has been employed in the past as a way to develop asymptotic expansions of HKs
has been the derivation of an integral equation for the HK, alternative to the differential equation~\eqref{eq:heat-equation}; 
see for example~\cite{Gaveau:1986, Barvinsky:1987uw, Bordag:1999ed, Bordag:2001ta, Bordag:2001fj}.
Customarily, such an integral equation would involve the free HK;
instead, in the present case $\Kzeta$ will play its role.

\begin{proposition}\label{prop:integral_equation}
The HK associated to $\ope$  satisfies the integral equation
\begin{align}\label{eq:integral_equation}
 \begin{split}
&K(x,y;T;\ope)= \Kzeta(x,y;T)+ \zeta\int^T_0 \dx[s] \int_{\mathbb{R}}\dx[z]\, \Kzeta(x,z;T-s) V(z) K(z,y;s;\ope).
 \end{split}
\end{align}
\end{proposition}
\begin{proof}
 The proof follows from direct application of the operator $\ope$ on both sides of Eq.~\eqref{eq:integral_equation}, analogously to the usual case. 
 After using the corresponding heat equation for $\Kzeta$, one can interpret $V(z) K(z,y;s;\ope)=[\partial_s-\partial_z^2+\lambda \delta(z-L)]K(z,y;s;\ope)$.
 Using the symmetry of $\Kzeta$ in its first two arguments after integrating by parts in $s$ and $z$,  one obtains then the desired result.
\end{proof}

The benefit of introducing the integral equation \eqref{eq:integral_equation} is that it is particularly well-suited to perform a perturbative expansion in the potential. 
Indeed, by direct application of the heat equation and subsequent integrations by parts in the intermediate propertimes, 
one can prove  the following corollary.
\begin{corollary}\label{cor:perturbative}
The HK of the operator $\ope$ can be expanded as a\old{formal}series in $\zeta$,
\begin{align}\label{eq:perturbative_expansion}
 \begin{split}
K(x,y;T;\ope)
 &=\sum_{n=0}^{\infty} K^{(n)}(x,y;T;\ope)\zeta^n,
 \end{split}
\end{align}
where the coefficients are made of iterated integrals of the potential $V$ and the heat-kernel $\Kzeta$:
\begin{align}\label{eq:perturbative_solution}
 \begin{split}
 K^{(n)}(x,y;T;\ope):&= \int_0^T \dx[s_n]\cdots \int^{s_2}_0 \dx[s_1] \int_{\mathbb{R}}\dx[z_1]\cdots \int_{\mathbb{R}}\dx[z_n]\, \Kzeta(x,z_1;s_1) 
  \\
  &\hspace{0.6cm}\times V(z_1) \Kzeta(z_1,z_2;s_2-s_1) \cdots V(z_n)   \Kzeta(z_n,y;T-s_n)
   \end{split}
\end{align}
and $K^{(0)}(x,y;T;\ope):=\Kzeta(x,y;T)$.
\end{corollary}
\old{
In accordance with the hypotheses detailed in Sec.~\ref{sec:introduction}, 
all the integrals in Eq.~\eqref{eq:perturbative_solution} exist as long as the potential $V$ decays fast enough at infinity.
This comment implies that this expansion is not valid for potentials that diverge at infinity, such as a harmonic one. 
One general situation in which this expansion is expected to be useful is that in which the potentials have two scales, 
one associated with its amplitude and the other one with its derivatives, being the latter much bigger than the former.
}

\section{Heat-kernel's trace: first order in $V$}\label{sec:first_order}
We will now show how to obtain a closed expression for the HK's trace at linear order in $V$.
Let us first introduce, as usual, a smearing function $f(x)\in C_{0}^{\infty}(\mathbb{R})$, 
i.e. an infinitely differentiable function with compact support.
Recall also that the imaginary error function $\operatorname{erfi}(\cdot)$ and the complementary error function $\operatorname{erfc}(\cdot)$
are defined in terms of the error function $\operatorname{erf}(\cdot)$ as
\begin{align}
\operatorname{erf}(x):=\frac{2}{\sqrt{\pi}} \int_0^x \dx[t]\, e^{-t^2},\quad \operatorname{erfi}(x):= -\mathi \operatorname{erf}(\mathi x),\quad 
\operatorname{erfc}(x):= 1- \operatorname{erf}(x).
\end{align}

\begin{theorem}\label{th:firs_order}
 The trace of the operator $\ope$'s HK, under the assumptions around Eq.~\eqref{eq:operator_D} and at linear order in $V$, is given by 
\begin{align}
 \begin{split}\label{eq:order1_hk_f}
  \Tr\left( f(x) K^{(1)}(x,y;T;\ope) \right) &= \int_\mathbb{R} \dx \,f(x) \hkf_{1,M}^{(f)}(\der_x;T) V(x)
  \\
  &\hu+ \int_{\Sigma}\dxd[x][0]  \hkf_{1,\Sigma}^{(f)}(\der_y,\der_z;T;\lambda) f(y) V(z)\Big\vert_{y=z=x},
 \end{split}
\end{align}
where we have defined the kernels
\begin{align}
 \begin{split}\label{eq:order1_hk_f_bulk}
 \hkf_{1,M}^{(f)}(k;T)
  :&=
  \frac{   e^{-\frac{1}{4} \left(k^2 T\right)}  \operatorname{erfi}\left(\frac{\sqrt{T  k^2} }{2}\right)}{2 \sqrt{ k^2} },
 \end{split}
\\
\begin{split}\label{eq:order1_hk_f_boundary}
\hkf_{1,\Sigma}^{(f)}(k_1,k_2;T;\lambda)
:&= 
\Bigg[ -\lambda\frac{k_1 \left(\left(k_1+k_2\right) \lambda ^2+k_1^2 \left(k_1-k_2\right)\right) }{ \left(k_1^2-k_2^2\right) k_2  \left(\lambda ^2+k_1^2\right)^2} e^{-\frac{1}{4} k_1^2 T} \operatorname{erfi}\left(\frac{k_1 \sqrt{T}}{2}\right)
\\
&\hspace{-2cm}
-\frac{\lambda^2  k_1 \left(\lambda ^2+k_1^2-2 k_1 k_2\right) e^{-\frac{1}{4} k_1^2 T}}{ \left(k_1^2-k_2^2\right) k_2  \left(\lambda ^2+k_1^2\right)^2}
+ \{k_1\leftrightarrow k_2\}
\Bigg]
\\
&\hspace{-2cm}+\frac{\lambda\left(k_1+k_2\right) }{ k_1 k_2 \left(\lambda ^2+\left(k_1+k_2\right)^2\right)}e^{-\frac{1}{4} \left(k_1+k_2\right)^2 T} \operatorname{erfi}\left(\frac{1}{2} \left(k_1+k_2\right) \sqrt{T}\right)
\\
&\hspace{-2cm}+\frac{\lambda^2  e^{\frac{\lambda ^2 T}{4}} \operatorname{erfc}\left(\frac{\lambda  \sqrt{T}}{2}\right)}{2 \left(\lambda ^2+k_1^2\right)^2 \left(\lambda ^2+k_2^2\right)^2 \left(\lambda ^2+\left(k_1+k_2\right)^2\right)} 
\Bigg\{2 \left(k_1^2+k_2 k_1+k_2^2\right) \lambda ^6 T
\\
&\hspace{-1.cm}+\lambda ^4 \Big[\left(k_1^2+k_2 k_1+k_2^2\right)^2 T-10 k_1 k_2\Big]+k_1 k_2 \lambda ^2 \Big[k_1 k_2 \left(k_1+k_2\right)^2 T
\\
&\hspace{-0.0cm}-2 \left(k_1^2-4 k_2 k_1+k_2^2\right)\Big]+2 k_1^2 k_2^2 \left(2 k_1^2+3 k_2 k_1+2 k_2^2\right)+\lambda ^8 T\Bigg\}
\\
&\hspace{-2cm}-\frac{\lambda ^3 \sqrt{T}}{ \sqrt{\pi } \left(\lambda ^2+k_1^2\right) \left(\lambda ^2+k_2^2\right)}
+\frac{\lambda^2  e^{-\frac{1}{4} \left(k_1+k_2\right)^2 T}}{ k_1 k_2 \left[\lambda ^2+(k_1+ k_2)^2\right]}.
\end{split}
\end{align}
\end{theorem}
\begin{proof}
A direct way to prove this theorem is to consider Corollary~\ref{cor:perturbative} at linear order in $V$. 
 For the bulk contribution we have performed 
all the necessary integrations by part in order to remove all the derivatives acting on $f$; 
this is possible because of the assumed properties of $f$ and $V$.
Of course, a similar procedure can not be implemented for the boundary terms.
\end{proof}

Some comments are in order. 
First,
the bulk contributions are independent of the delta potential. 
Indeed, if we expand for small coupling $\ampli$, we get
\begin{align}\label{eq:order1_free}
 \begin{split}
\Tr\left( f(x) \Kope^{(1)}(x,y;T) \right)
 &= \int_{\mathbb{R}} \dx  f(x) \frac{ e^{\frac{\partial_x^2 T}{4}} \operatorname{erfi}\left(\frac{\sqrt{-T \partial_x^2 } }{2}\right) } { 2 \sqrt{-\partial_x^2}} V(x)+O\left(\lambda ^1\right) 
 \\
 &\hspace{-2cm}= \int_{\mathbb{R}} \dx f(x) \frac{\sqrt{T}}{4}\left( \sum_{n=0}^{\infty} \frac{1}{\Gamma\left(n+\frac{3}{2}\right)} \left(\frac{{T\partial_x^2}}{4}\right)^{n}\right)  V(x)+O\left(\lambda ^1\right).
 \end{split}
 \end{align}
These are standard formulae in the literature. 
As an example, the small-propertime ($T$) expansion 
in the second line is consistent with the results contained in~\cite{Vassilevich:2003xt} for the GSDW coefficients up to $a_6$;
notice that, order by order in $T$, it depends only on integer powers of $\partial_x^2$, i.e. it is made of local contributions.
Moreover, the formulae in Eq.~\eqref{eq:order1_free} are compatible with the second order resummed result à la Barvisnky--Vilkovisky~\cite{Barvinsky:1987uw}, 
see \cite{Vassilevich:2003xt} and the discussion in Sec.~\ref{sec:second_order} of the present manuscript.
 
Second, even if the presence of rational functions of the $k_i$ may induce one to think that nonlocalities may be present
in the small-T expansion of the surface contributions,
an explicit calculation shows that all the corresponding coefficients are made of local invariants.
Even if the computations are lengthy, this can be straightforwardly seen from the following corollary 
of Theorem~\ref{th:firs_order}.
\begin{corollary}\label{cor:order1_SDW_coef}
 The surface contributions to the GSDW coefficients associated to the operator $\ope$, at linear order in $V$, 
vanish for $n<4$. For $n\geq 4$ an even number, they read
\begin{align}\label{eq:order1_SDW}
 \begin{split}
&a_{n,\Sigma}^{(1)}(f;\ope)
=
\Bigg\lbrace
\frac{ 2^{1-n}\lambda}{\Gamma \left(\frac{n+1}{2}\right)}
\Bigg[
-\frac{(-1)^{ n} (n-1) \lambda ^{n}}{\left(\lambda ^2-\partial _1^2\right) \left(\lambda ^2-\partial _2^2\right)}
+\frac{\left(\partial _1+\partial _2\right){}^{n}}{\partial _1 \partial _2 \left(\lambda ^2-(\partial _1+\partial _2)^2\right)} \\
 &
 \hu+\frac{(-1)^{ n} \partial _1 \partial _2 \Big[-5 \lambda ^4+\lambda ^2 \partial _2^2+2 \partial _1 \partial _2 \left(\partial _2^2-2 \lambda ^2\right)+\partial _1^2 \left(\lambda ^2+3 \partial _2^2\right)+2 \partial _1^3 \partial _2\Big] \lambda ^{n}}{\left(\lambda ^2-\partial _1^2\right){}^2  \left(\lambda^2- (\partial _1+\partial _2)^2\right) \left(\lambda ^2-\partial _2^2\right){}^2}
 \\
 &
 \hu+\Bigg(
 \frac{\left(-\lambda ^2 \partial _1-\lambda ^2 \partial _2+\partial _1^3-\partial _2 \partial _1^2\right) \partial _1^{n}}{\partial _2 \left(\partial _1^2-\partial _2^2\right) \left(\lambda ^2-\partial _1^2\right){}^2}
 +\{\partial_1\leftrightarrow \partial_2 \}
 \Bigg)
 \Bigg]f(x_1)V(x_2)\Bigg\vert_{x_i=x}
 \Bigg\rbrace[\Sigma]
 ,
 \end{split}
 \end{align}
 while, for $n\geq 5$ and odd, we have
 \begin{align}\label{eq:order1_SDW2}
 \begin{split}
& a_{n,\Sigma}^{(1)}(f;\ope)
 =\Bigg\lbrace
 \frac{2^{1-n} \lambda ^2}{\Gamma \left(\frac{n+1}{2}\right)}
 \Bigg[
 \Bigg( \frac{\left(\lambda ^2-\partial _1^2+2 \partial _1 \partial _2\right) \partial _1^{n}}{\partial _2 \left(\partial _1^2-\partial _2^2\right) \left(\lambda ^2-\partial _1^2\right){}^2}
 + \{\partial_1\leftrightarrow \partial_2\}
 \Bigg)
 \\
 &
 \hu-\frac{\partial _1 \partial _2 \Big[-5 \lambda ^4+\lambda ^2 \partial _2^2+2 \partial _1 \partial _2 \left(\partial _2^2-2 \lambda ^2\right)+\partial _1^2 \left(\lambda ^2+3 \partial _2^2\right)+2 \partial _1^3 \partial _2\Big] \lambda ^{n-1}}{\left(\lambda ^2-\partial _1^2\right){}^2 \left(\lambda ^2-\partial _2^2\right){}^2 \left(\lambda ^2-\left(\partial _1+\partial _2\right){}^2\right)}
 \\
 &
 \hu+\frac{\left(\partial _1+\partial _2\right){}^{n-1}}{\partial _1 \partial _2 \left(-\lambda^2 + (\partial _1+\partial _2)^2\right)}
 +\frac{(n-1) \lambda ^{n-1}}{\left(\lambda ^2-\partial _1^2\right) \left(\lambda ^2-\partial _2^2\right)}
 \Bigg]
 f(x_1)V(x_2)\Bigg\vert_{x_i=x}
 \Bigg\rbrace[\Sigma].
\end{split}
\end{align}
\end{corollary}

\begin{comm}\label{comm:order1}
 \old{ 
 The discussion in Remark~\ref{comm:order0} trivially extends to this order,
  with the addition that the coefficients will in general possess a functional dependence on derivatives of $V$.
  }
\end{comm}

Coming back to Eq.~\eqref{eq:order1_hk_f_boundary},  one can corroborate it by comparing with several previously known results.
To begin with, if we consider a constant potential $V\equiv V_0$, then we obtain
\begin{align}
 \Tr\left( f(x) K^{(1)}(x,y;T;\ope) \right)
 &\overset{V\equiv V_0}{=}  T  V_0 \Tr\Big( \Kzeta(x,y;T) \Big),
 \end{align}
which is the correct result since in that case\old{we know the all-order expression $K_{V\equiv V_0}(\ope)= e^{\zeta TV_0} \Kzeta$}.
As an additional check, we can also remove the smearing function. 
In such a situation the result greatly simplifies.
\begin{corollary}
 Provided that the potential $V$ and its derivatives of any order decay sufficiently fast at infinity, 
 we can set the smearing function to unity in Eq.~\eqref{eq:order1_hk_f} and obtain
\begin{align}\label{eq:order1_no_smear}
 \begin{split}
\Tr\left( K^{(1)}(x,y;T;\ope) \right) &=\sqrt{\frac{T}{4\pi}}\int_{\mathbb{R}} \dx \,V(x) - \frac{\lambda  T }{2}\int_{\Sigma} \dxd[x][0] \frac{e^{\frac{ T\partial_x^2}{4}} }{   \left(-\partial_x^2+\lambda ^2\right)} 
 \\
 &\hspace{-2cm}\times \left\{\lambda   \left[1 -e^{\frac{1}{4} T \left(-\partial_x^2+\lambda ^2\right)} \operatorname{erfc}\left(\frac{\lambda  \sqrt{T}}{2}\right) \right]-\operatorname{erf}\left(\frac{\sqrt{T}\partial_x }{2}\right)\partial_x \right\}  V(x).
 \end{split}
\end{align}
In this case, the GSDW coefficients listed in Corollary~\ref{cor:order1_SDW_coef}, nonvanishing only for $n\geq 4$, simplify to
\begin{align}\label{eq:order1_coefficients_nof}
 a^{(1)}_{n,\Sigma}(1;\ope)&=
 \frac{2^{2-n}}{   \Gamma \left(\frac{n-1}{2}\right)}
 \int_{\Sigma}
 \frac{\dxd[x][0]}{\left(-\partial_x^2+\lambda ^2\right)}
 \begin{cases}
         -\Big( \lambda ^{n-1}+\lambda  \mathi^n (-\partial_x^2)^{(n-2)/2}\Big)  V,\; \text{for $n$ even},\\
         \Big(\lambda ^{n-1}+\lambda ^2 \mathi^{n-1} (-\partial_x^2)^{(n-3)/2}\Big) V,\;\text{for $n$ odd}.
        \end{cases}
\end{align}
\end{corollary}
Notice that, as previously, the denominator in Eq.~\eqref{eq:order1_coefficients_nof} is formal: 
it is understood that one should first perform the division of the polynomials (which is exact for all $n$) before interpreting $\partial_x^2$ as a differential operator.
The GSDW coefficients up to $n=5$ coincide with the corresponding ones in Ref.~\cite{Bordag:1999ed}, 
while the $a_6$ term agrees with the result in Ref.~\cite{Vinas:2010ix}.

\old{One further comparison can be done employing the large $\lambda$ expansion;
we will postpone this analysis to Sec.~\ref{sec:bcs}.}

\section{Heat-kernel's trace: second order in $V$}\label{sec:second_order}
Increasing the order in the potential substantially increases the difficulty in the computation of the HK's trace.
To simplify the results to quadratic order in $V$ and
taking also into account that in physical applications it is customary to do so,
in the following we will set the smearing function aside.

\begin{theorem}\label{th:quadratic}
 The trace of the operator $\ope$'s HK, under the assumptions around Eq.~\eqref{eq:operator_D}, at quadratic order in $V$ and neglecting total derivatives, is given by 
\begin{align}
 \begin{split}\label{eq:order2_hk_f}
  \Tr\left(  K^{(2)}(x,y;T;\ope) \right) &= \int_\mathbb{R} \dx\, V(x) \hkf_{2,M}(\der_x;T) V(x)
  \\
  &\hu+ \int_{\Sigma}\dxd[x][0]  \hkf_{2,\Sigma}(\der_1,\der_2;T;\lambda) V(x_1) V(x_2)\Big\vert_{x_i=x},
 \end{split}
\end{align}
where the kernels are related to those present in the linear expansion:
\begin{align}\label{eq:relation_kernels}
 \hkf_{2,M}(k;T) =\frac{T}{2} \hkf_{1,M}^{(f)}(k ;T),\quad  \hkf_{2,\Sigma}(k_1,k_2;T;\lambda) = \frac{T}{2}\hkf_{1,\Sigma}^{(f)}(k_1,k_2;T;\lambda).
\end{align}
Correspondingly, the GSDW coefficients read
\begin{align}
 a_{n,\Sigma}^{(2)}(1;\ope)= \frac{1}{2} a_{n-2,\Sigma}^{(1)}(V;\ope).
\end{align}

\end{theorem}

\begin{proof}
The correctness of Theorem~\ref{th:quadratic} can be shown by appealing to the perturbative expansion in Eq.~\eqref{eq:perturbative_solution}.
At quadratic order in $V$ we notice that the only relevant variable is $s_2-s_1$; changing variables we thus get
\begin{align}
 \begin{split}
  &\Tr\left(  K^{(2)}(x,y;T;\ope) \right) =\int_0^T \dx[s_2]\int^{s_2}_0 \dx[s_1] 
  \\
  &\hu \times \int_{\mathbb{R}^2}\dx[z_1]\dx[z_2]\,  V(z_1) \Kzeta(z_1,z_2;s_2-s_1) V(z_2)   \Kzeta(z_2,z_1;T-s_2+s_1)
  \\
  &=\int_0^T \dx[s_2] \int_0^{s_2} \dx[s_-] 
  \int_{\mathbb{R}^2}\dx[z_1]\dx[z_2]\,  V(z_1) \Kzeta(z_1,z_2;s_-) V(z_2)   \Kzeta(z_2,z_1;T-s_-)
  \\
  &= \frac{T}{2} \int_0^T \dx[s_-]  \int_{\mathbb{R}^2}\dx[z_1]\dx[z_2]\,  V(z_1) \Kzeta(z_1,z_2;s_-) V(z_2)   \Kzeta(z_2,z_1;T-s_-).
\end{split}
\end{align}
After replacing $V(z_1)\to f(z_1)$, this is proportional to the linear order expression with a smearing function, i.e. proportional to $\Tr\left( f K^{(1)}(\ope)\right)$.
The fact that we are neglecting total derivatives is a consequence of the 
integration by parts that we have performed in obtaining the kernels of the linear expression,
cf. the proof of Theorem~\ref{th:firs_order}. 
Finally, the relation between the GSDW coefficients follows directly from the second equality in Eq.~\eqref{eq:relation_kernels}.
\end{proof}
\begin{comm}\label{comm:order2}
 The discussion in Remark~\ref{comm:order1} trivially extends to this order.
\end{comm}

Theorem~\ref{th:quadratic} clarifies the comment made after Eq.~\eqref{eq:order1_free}:
it can be easily checked that $\hkf_{2,M}$ is the flat-space version of the Barvinsky--Vilkovisky kernel for the quadratic contribution in $V$, 
which is proportional by a factor $T/2$ to $\hkf_{1,M}^{(f)}$. 
Alternatively, one can perform a small-propertime expansion and see that the first terms (up to order $T^3$) 
agree with the contributions listed in~\cite{Vassilevich:2003xt}, of course after neglecting boundary terms.

An additional corroboration can be done by noting that in the limit when both $V$ factors becomes a constant $V_0$ we get
\begin{align}
 \Tr\left(  K^{(2)}(x,y;T;\ope) \right) &\overset{V\equiv V_0}{=}\frac{T^{2} V_0^2}{2}
   \Tr \left( \Kzeta(x,y;T)\right),
\end{align}
which again matches the expansion of $e^{\zeta T V_0 } \Kzeta^{}$, this time at quadratic order in $V_0$.
Instead, if just one of the $V$ becomes constant, we get the equivalent of Eq.~\eqref{eq:order1_no_smear}:
\begin{align}
\begin{split}
 \Tr\left(  K^{(2)}(x,y;T;\ope) \right) &=\frac{TV_0}{2} \Tr\left(  K^{(1)}(x,y;T;\ope) \right).
   \end{split}
\end{align}

As a matter of completeness, 
let us perform an expansion for small propertime;\old{the one for large coupling will be analysed 
in Sec.~\ref{sec:bcs}.}As a result we get the following surface contributions,
\begin{align}
\begin{split}
 &\Tr\left(  K^{(2)}(x,y;T;\ope) \right)\Big\vert_{\rm surface} 
 \\
 &=  \int_\Sigma \dxd[x][0] \Bigg[ -\frac{\lambda  }{4 \sqrt{\pi }} T^{5/2}+\frac{1}{16} \lambda ^2 T^{3}- \frac{\lambda }{24 \sqrt{\pi }}\left({\lambda ^2}{}+2 {\partial_1^2  }+{\partial_1 \partial_2 }\right)T^{7/2}
  \\
  &\hspace{2.0cm}+\frac{ \lambda ^2}{128} \left({\lambda ^2}+2 \partial_1^2  +\frac{1}{6} \partial_1 \partial_2 \right) T^{4}+\mathcal{O}(T^{9/2})\Bigg] V(x_1) V(x_2)\Bigg\vert_{x_i=x},
\end{split}
\end{align}
whose first term can be compared with the result in Ref.~\cite{Vinas:2010ix}.

\section{A connection with Dirichlet and Robin boundary conditions}\label{sec:bcs}

\subsection{Dirichlet boundary conditions}
\old{One further check of all the previous results can be performed considering the large-$\lambda$ expansion.
The first boundary contributions will be independent of $\lambda$ 
and are expected to be related to the HK of a free Laplacian with Dirichlet boundary conditions.
Indeed, it is well-known that in such a limit the operator $\ope_{\zeta=0}$ acts as a free Laplacian on two half-spaces: 
they share $\Sigma$ as boundary, on which Dirichlet boundary conditions are imposed. 
This is to be expected on physical grounds, since as $\lambda\to \infty$ the layer $\Sigma$ foreseeable becomes impenetrable. 
}

\old{
In the $\lambda\to\infty$, one can obtain the following kernels, 
which define the trace of the heat-kernel up to quadratic order in the potential $V$:
\begin{align}\label{eq:dirichlet0}
\hkf_{0,\Sigma}^{(f)}(k_1,k_2;T;\lambda\to\infty)
&=-\frac{1}{2} e^{-\frac{k^2 T}{4} },
\\
\begin{split}\label{eq:dirichlet1}
\hkf_{1,\Sigma}^{(f)}(k_1,k_2;T;\lambda\to\infty)&
\\
&\hspace{-2cm}=\left[ \frac{e^{-\frac{1}{4} \left(k_1+k_2\right){}^2 T} k_2^2 }{k_1 \left(k_1-k_2\right) k_2 \left(k_1+k_2\right)} \left(e^{\frac{1}{4} k_1 \left(k_1+2 k_2\right) T}-1\right)+ \{k_1\leftrightarrow k_2\}\right],
\end{split}
\\\label{eq:dirichlet2}
\hkf_{2,\Sigma}(k_1,k_2;T;\lambda\to\infty) 
&= \frac{T}{2}\hkf_{1,\Sigma}^{(f)}(k_1,k_2;T;\lambda\to\infty).
\end{align}
}

\old{
At first order in the potential, the expression simplifies if we take the smearing function as the unit function;
in that case, readable expressions are obtained for the surface contributions even at higher orders in $\lambda^{-1}$:
\begin{align}
\begin{split}\label{eq:order1_large_coupling}
 &\Tr\left( K^{(1)}(x,y;T;\ope) \right)\Big\vert_{\text{surface}} = 
  \int_{\Sigma}\dxd[x][0]
  \Bigg\lbrace -\frac{T e^{\frac{ T \partial_x^2}{4}}}{2}
 \\
 &\hu+ \left[\frac{\sqrt{T} }{ \sqrt{\pi}}+\frac{ T  e^{\frac{ T\partial_x^2}{4}}}{2 } \operatorname{erf}\left(\frac{ \sqrt{T} \partial_x }{2}\right) \partial_x  \right] \lambda^{-1}
 -\frac{ T e^{\frac{ T \partial_x^2}{4}}\partial_x^2}{2  } \lambda ^{-2}+\mathcal{O}(\lambda^{-3})
 \Bigg\rbrace V(x)
 \\
 &=- \frac{T}{2}\int_\Sigma \dxd[x][0] \Bigg[\sum_{j=0}^{\infty} \frac{T^j (\partial_x^2)^j}{4^j j!}V(x)\Bigg]+O\left(\lambda^{-1}\right).
 \end{split}
\end{align}
Taking the previous comments into account, the correctness of the coefficients multiplying the $V(L)$ and 
$\partial^2 V(L)$ factors can be respectively  verified by comparing with the Dirichlet results in Refs.~\cite{Vassilevich:2003xt} and \cite{Branson:1995cm}.
}

\old{
On the other side, at second order in $V$ we obtain that the surface terms are given by
\begin{align}
\begin{split}
 &\Tr\left(  K^{(2)}(x,y;T;\ope) \right)\Big\vert_{\rm surface} 
 \\
 &= \int_\Sigma \dxd[x][0]\Bigg\lbrace 
 \hkf_{2,\Sigma}(k_1,k_2;T;\lambda\to\infty) +
\\
&\hspace{2.0cm}\Bigg[ -\frac{\mathi T^{} \left[-2\partial_1^2 F\left(\frac{-\mathi \sqrt{T} \partial_1}{2}\right)+\left(-\partial_2^2+\partial_1^2\right) F\left(\frac{-\mathi}{2} \sqrt{T} \left(\partial_1+\partial_2\right)\right)\right]}{2\sqrt{\pi } \partial_1 \left(\partial_1-\partial_2\right) \partial_2} \lambda ^{-1}
 \\
 &\hspace{6.5cm}+ \{\partial_1\leftrightarrow \partial_2\}\Bigg]+\mathcal{O}(\lambda^{-2})\Bigg\rbrace
 V(x_1)V(x_2) 
 \\
 &=T^{-1/2}\int_\Sigma \dxd[x][0] 
 \Bigg\lbrace 
 -\frac{T^{5/2}}{4}-\frac{\left(8 \partial_1^2+5 \partial_2 \partial_1\right)}{64}  T^{7/2}-\frac{\left(3 \partial_1^4+7 \partial_2 \partial_1^3+5 \partial_2^2 \partial_1^2\right)}{192}  T^{9/2}
 \\
 &
 \\
 &-\frac{\left(16 \partial_1^6+54 \partial_2 \partial_1^5+112 \partial_2^2 \partial_1^4+69 \partial_2^3 \partial_1^3\right)}{12288}  T^{11/2}
 +\mathcal{O}(\lambda^{-1},T^{6})
\Bigg\rbrace V(x_1)V(x_2) \Bigg\vert_{x_i=x},
 \end{split}
 \end{align}
 where $F(\cdot)$ is Dawson's integral, $F(x):=e^{-x^2}\int_0^xe^{y^2}\dx[y]$.
In the third line, we have rendered explicit the small-propertime expansion of the Dirichlet limit; 
one can readily see that the coefficient accompanying $\int_\Sigma \dxd[x][0]\, V^2(x)$ agrees with the expression in Ref.~\cite{Branson:1995cm}.
}

\old{
Notice that in all these expansions there appear no terms with an odd total number of derivatives.
In fact, they simply vanish in our case, 
because of a cancellation between contributions from the two different semi-spaces:
this is a consequence of the fact that the outward normal-pointing vector has different sign on the two sides of $\Sigma$. 
}

\old{
More generally, we can write down the expressions of the GSDW coefficients in the $\lambda\to\infty$ limit.
For $n$ even they vanish, what can be simply proven by dimensional arguments, taking into account the \hidepdo{preceding}{precedent} paragraph. 
For $n$ odd they are given by
\begin{align}\label{eq:dirichlet_SDW}
 \begin{split}
a_{n,\Sigma}^{(0)}(f;\ope_{\lambda\to\infty})
&=
\Bigg\lbrace
\Bigg[
-\frac{2^{-n} (\partial)^{n-1} }{\Gamma \left(\frac{n+1}{2}\right)}
 \Bigg]f(x)
 \Bigg\rbrace[\Sigma]
 ,
 \\
 a_{n,\Sigma}^{(1)}(f;\ope_{\lambda\to\infty})
&=
\Bigg\lbrace
\Bigg[
\frac{ \left(\partial _2-\partial _1\right) \left(\partial _1+\partial _2\right){}^n+\partial _1^{n+1}-\partial _2^{n+1}}{2^{n-1} \partial _1 \partial _2 \left(\partial _1^2-\partial _2^2\right) \Gamma \left(\frac{n+1}{2}\right)}
\Bigg]f(x_1)V(x_2)\Bigg\vert_{x_i=x}
 \Bigg\rbrace[\Sigma]
 ,
 \\
 a_{n,\Sigma}^{(2)}(1;\ope_{\lambda\to\infty})
&=
 \frac{1}{2}a_{n-2,\Sigma}^{(1)}(V;\ope_{\lambda\to\infty})
 .
 \end{split}
 \end{align}
 One can rephrase this result in the following way.
 Consider a smooth manifold $M_0$ with boundary $\partial M_0=\Sigma$.
 Let us call $a_{n,\partial M_0}(f,\ope_{\lambda=0},\mathcal{B}_D)$ the boundary contributions to  the $n$th HK coefficient
 of the operator $\ope_{\lambda=0}$ restricted to $M_0$, with Dirichlet boundary conditions
 imposed on $\partial M_0$:
 \begin{align}\label{eq:dirichlet_bc}
  \mathcal{B}_D \phi =\phi\vert_{\partial M_0}=0.
 \end{align}
 Then, the coefficients satisfy
\begin{align}\label{eq:delta_dirichlet}
 a_{n,\partial M_0}(f,\ope_{\lambda=0},\mathcal{B}_D) &=\frac{1}{2}a_{n,\Sigma}(f,\ope_{\lambda\to\infty})+\cdots,
\end{align}
where the dots denote terms with a functional structure different from the ones already present in $a_{n,\Sigma}(f,\ope_{\lambda\to\infty})$.
}

\subsection{Robin boundary conditions}
\old{
We will dedicate this section to an elegant argument that links some heat-kernel coefficients for semitransparent conditions 
with those corresponding to Dirichlet and Robin boundary conditions. 
}

\old{To be more explicit, following Ref.~\cite{Gilkey:2001mj}, 
consider again a smooth manifold $M_0$ with boundary $\partial M_0=\Sigma$.
Afterwards, glue two copies of the former along their common boundary, giving rise to the manifold $M$.
In this construction, take a Laplace operator $D$, defined on the glued space $M$, which satisfies the following two conditions: it
implies semitransparent boundary conditions on $\Sigma$ (with parameter $\lambda$)
and reduces to the Laplace operator $D_0$ in both copies of $M_0$. 
Additionally, consider a smearing function $f$ that is even upon reflections on $\Sigma$.
Under rather general assumptions, the formula proved in Ref.~\cite{Gilkey:2001mj} is that 
the surface contributions to the heat-kernel coefficients satisfy
\begin{align}\label{eq:delta_dirichlet_robin}
 a_{n,\Sigma}(f,D,\lambda) = a_{n,\partial M_0}(f_0,D_0,\mathcal{B}_D)+a_{n,\partial M_0}(f_0,D_0,\mathcal{B}_S),
\end{align}
where $\mathcal{B}_D$ imply Dirichlet boundary conditions for the operator $D_0$ and $\mathcal{B}_S$ the following Robin ones\footnote{\old{The derivative in the normal direction to $\partial M_0$ is denoted by $\partial_n$.}}:
\begin{align}
\quad \mathcal{B}_S \phi =(\partial_n+S)\phi\vert_{\partial M_0}=0,\quad S:=-\lambda/2.
\end{align}
}
\old{
This result is applicable to our particular case, as long as one considers $V(x)$ and $f(x)$ that are even under reflection across $\Sigma$.
One should notice that the latter restriction means that terms containing $\partial^n V$ or $\partial^n f$ evaluated on the boundary, with $n$ odd, 
will trivially vanish.
Taking this into account and considering the coefficients for Dirichlet and Robin boundary conditions listed in Ref.~\cite{Vassilevich:2003xt}, 
it is a straightforward task to check that our first HK coefficients, both at zeroth and first order in $V$, satisfy the relation~\eqref{eq:delta_dirichlet_robin}.
}

\old{
Furthermore, one can invert the argument and obtain an infinite number of terms corresponding to HKs with Robin boundary conditions.
Indeed, applying the result~\eqref{eq:delta_dirichlet_robin} to our operator $\ope$ (therefore assuming $V$ and $f$ even) and substituting the Dirichlet contribution with the help of Eq.~\eqref{eq:delta_dirichlet}, we obtain
\begin{align}\label{eq:delta_robin}
 a_{n,\partial M_0}(f,\ope_{\lambda=0},\mathcal{B}_S)=a_{n,\Sigma}(f,\ope,\lambda) -\frac{1}{2}a_{n,\Sigma}(f,\ope_{\lambda\to\infty})+\cdots,
\end{align}
where the coefficients in the RHS have been computed in the previous sections (up to second order in $V$) and the dots denote terms with a functional structure different from the ones already present in $a_{n,\Sigma}(f,\ope,\lambda)$ or $a_{n,\Sigma}(f,\ope_{\lambda\to\infty})$.
}
\section{Surface form factors in quantum field theory}\label{sec:qft}

As a simple, albeit conceptually rich application of the preceding results, let us now
consider a quantum scalar field in  $d$-dimensional flat Euclidean space. We will couple it quadratically to the background field $\sigma$, 
according to the following action:
\begin{align}\label{eq:action}
 S := \frac{1}{2}\int_{\mathbb{R}^d} \dxd[x][d] \bigg[ (\partial\phi)^2 +m^2 \phi^2 + \lambda \delta(\perpe[x]-L) \phi^2+ \phi^2 \sigma^2 \bigg].
\end{align}
For a general discussion of this type of theories with and without the boundary term see respectively Refs.~\cite{Bordag:2004rx,Toms:2006re,Toms:2012dc} and~\cite{Mazzitelli:2011st,Franchino-Vinas:2021lbl}.
Following the terminology in the QFT literature, the negative of the Laplacian  will be denoted $\Box:=-\partial^2$.
 As usually, 
the one-loop contribution to the effective action can be written in terms of the operator of quantum fluctuations~\cite{Itzykson},
\begin{align}\label{eq:qft_operator}
 A:=\Box +m^2 + \lambda \delta(\perpe[x]-L) +  \sigma^2 ,
\end{align}
either as a function of its determinant or, introducing an integral over Schwinger's propertime $T$,
in terms of its HK:
\begin{align}
 \begin{split}\label{eq:one_loop_EA}
\Gamma_{\rm 1-loop}&= \frac{1}{2} \operatorname{Log} \operatorname{Det} A
 = - \frac{1}{2} \int_0^\infty \frac{\dx[T] }{T} \operatorname{Tr} e^{-T A}.
 \end{split}
 \end{align}
\old{Recall that, as discussed in Sec.~\ref{sec:introduction}, in the following we will assume $\lambda\geq0$;
otherwise, the integrals in the propertime would be ill-defined.}

 Using the formulae for the HK developed in the previous sections, 
 we can recast the one-loop effective action as\footnote{The zeroth order in $\sigma$ can be read from Eq.~\eqref{eq:HK_delta_trace}.}
\old{
\begin{align}
  \begin{split}\label{eq:ea_form_factors}
&\Gamma_{\rm 1-loop}= \int_{\mathbb{R}^d} \dxd[x][d] \, \bigg[ F^{(0)}_{M}(m) + F^{(1)}_{M}(m) \sigma^2(x) + \sigma^2(x) F^{(2)}_{M}\big(\Box,m\big) \sigma^2(x)\bigg]
  \\
  &\hspace{0.5cm}+ \int_\Sigma \dxd[x][d-1] \bigg[ F^{(0)}_{\Sigma}(m,\lambda)+
  F^{(1)}_{\Sigma}\big(\der_{\perpe[y]},m,\lambda\big) \sigma^2(y)
  \\
  &\hspace{1.5cm}+
  F^{(2)}_{\Sigma}\big(\der_{\perpe[z]},\der_{\perpe[y]},m,\lambda\big) \sigma^2(z)\sigma^2(y)
  \bigg]_{z=y=x}+ \mathcal{O}\left(\sigma^6, \partial_{\parallel}\sigma\right),
 \end{split}
 \end{align}
 where}$F^{(i)}_{M}$ and $F^{(i)}_{\Sigma}$ are called form factors.
 In Eq.~\eqref{eq:ea_form_factors} we are neglecting powers of the field $\sigma$ higher than four, 
 as well as all possible contributions involving its partial derivatives with respect to the directions tangent to the surface $\Sigma$.
 For our purposes, this will turn out to be enough.
 
 A closed expression for the form factors in arbitrary dimensions can be obtained; 
 given that they are rather lengthy, we prefer to leave them to App. \ref{app:form_factors_d}.
 Instead, we will display here  the corresponding formulae in $d\equiv 4$ dimensions, 
 focusing on the terms that require to undergo a renormalization process; 
they are 
\begin{align}\label{eq:f0m}
F^{(0)}_{M}(m) &=
\frac{m^4}{128 \pi ^2 }\left[\frac{4}{(d-4)}+{2 \log \left(\frac{m^2}{4\pi\mu^2}\right)+2 \gamma -3}{} \right],
\\
\label{eq:f0s}
F^{(0)}_{\Sigma}(m,\lambda) &=\frac{\lambda(6 m^2-\lambda ^2)}{192 \pi ^2 }
\Bigg[ \frac{2}{{(d-4)}}+2\gamma-\frac{8}{3}+ \log \left(\frac{m^2}{\pi\mu ^2}\right) \Bigg] + C^{(0)}_\Sigma(m,\lambda),
\\
\label{eq:f1s}
F^{(1)}_{\Sigma}(k,m,\lambda) &=\frac{\lambda }{32 \pi ^2 }\left[\frac{2}{(d-4)}+\gamma +\log \left(\frac{m^2}{4 \pi  \mu^2}\right)\right] 
+
  C^{(1)}_\Sigma(k,m,\lambda),
\\
\label{eq:f1m}
F^{(1)}_{M}\big(m\big)&=\frac{m^2}{32 \pi ^2 } \left[\frac{2}{(d-4)}+\gamma -1+ {\log \left(\frac{m^2}{4 \pi  \mu^2}\right)}{}\right],
\\
 \begin{split}
\label{eq:f2m}
F^{(2)}_{M}\big(k^2,m\big)&=\frac{1}{64 \pi ^2}\left[\frac{2}{ (d-4)}+\gamma  +\log \left(\frac{m^2}{4\pi\mu^2}\right)\right]
+C^{(2)}_M(k^2,m),
\end{split}
\end{align}
where $\gamma$ is the Euler--Mascheroni constant\old{and $\mu$ is an arbitrary constant with units of mass (introduced to render the argument of the logarithm in Eq.~\eqref{eq:one_loop_EA} dimensionless).}Additionally,
we have used dimensional regularization\footnote{In this simple example, it proves convenient to introduce dimensional regularization by just modifying the HK's leading power of the propertime to be proportional to $(4\pi T)^{-d/2}$.
Alternatively, one could modify the power of $T$ in the denominator of Eq.~\eqref{eq:one_loop_EA} to be $T^s$ and take afterwards the limit $s\to1$.} and defined the functions
\begin{align}
  C^{(0)}_\Sigma(m,\lambda):&=4 \pi  m^3-\lambda  m^2- \left(\lambda ^2-4 m^2\right)^{3/2} \operatorname{arctanh} \left(\sqrt{1-\frac{4 m^2}{\lambda ^2}}\right),
\\
\begin{split}
 C^{(1)}_\Sigma(k,m,\lambda):&=\frac{\lambda}{32 \pi ^2 \left(k^2+\lambda ^2\right)}\Bigg\lbrace 
   2 \lambda  \sqrt{\lambda ^2-4 m^2} \operatorname{arccoth}\left(\frac{\lambda }{\sqrt{\lambda ^2-4 m^2}}\right)
    \\
   &\hspace{-1.5cm} +\pi  \lambda  \sqrt{k^2+4 m^2}+2 \sqrt{k^2 \left(k^2+4 m^2\right)} \operatorname{arctanh}\left(\sqrt{\frac{k^2}{k^2+4 m^2}}\right)  
 \Bigg\rbrace-\old{\frac{\lambda }{16 \pi ^2}},
 \end{split}
 \\
 C^{(2)}_M(k^2,m):&=\frac{1}{32 \pi ^2} 
\left[ {\left(\frac{4 m^2}{k^2}+1\right)}^{1/2}\, \operatorname{arcsinh}\left( \sqrt{\frac{k^2}{4m^2}}\right) -1 \right].
\end{align}

\subsection{The renormalization}\label{sec:renormalization}

The renormalization process follows now by absorbing the infinities of the theory into the dressed coupling constants. 
In the minimal subtraction (MS) scheme, one would just introduce counterterms to cancel the negative powers of $(d-4)$;
this is the path followed for example in Ref.~\cite{Bordag:2004rx}.\old{Explicitly, 
one should enlarge the Lagrangian of departure to include also terms with functional dependence equal to those of the necessary counterterms~\cite{Peskin:1995ev},
 \begin{align}\label{eq:general_action}
 \begin{split}
S_0 :&= \frac{1}{2}\int_{\mathbb{R}^4} \dxd[x][4] \bigg[ (\partial\phi)^2 +m^2 \phi^2 + \phi^2 \sigma^2 +a_{M,0}^{(0)}+a_{M,0}^{(1)} \sigma^2 +a_{M,0}^{(2)} \sigma^4\bigg]
\\
&\hu+\frac{1}{2}\int_{\Sigma} \dxd[x][3] \bigg[ \lambda  \phi^2 +a_{\Sigma,0}^{(0)} +a_{\Sigma,0}^{(1)} \sigma^2 \bigg].
\end{split}
\end{align}
The coefficients $\delta a_{M,0}^{(0)}$ and $\delta a_{\Sigma,0}^{(0)}$ correspond to the cosmological constant and the surface tension of $\Sigma$ (both up to some proportional factor),
while $\delta a_{M,0}^{(1)}$, $\delta a_{M,0}^{(2)}$ and $\delta a_{\Sigma,0}^{(1)}$  are the mass of the field $\sigma$, the coupling of a quartic self-interaction for $\sigma$
and a coupling of $\sigma$ to the boundary.
Notice that, in dimensional regularization and four dimensions, there would appear no further terms stemming from the neglected factors which are $\mathcal{O}(\sigma^6,\partial_{\parallel}\sigma)$; 
this can be seen by dimensional counting or by direct computation~\cite{Bordag:2004rx}. 
However, some further terms may be needed in other regularization schemes.
}

\old{
In dimensional regularization and a prescription in which we absorb all the $k$ independent one-loop corrections, 
the corresponding counterterms can be read from the Eqs.~\eqref{eq:f0m}, \eqref{eq:f0s}, \eqref{eq:f1s}, \eqref{eq:f1m} and \eqref{eq:f2m},
to wit:
\begin{align}\label{eq:counterterms_nonminimal}
\begin{split}
\delta a^{(0)}_{M} &=
-\frac{m^4}{64 \pi ^2 }\left[\frac{4}{(d-4)}+{2 \log \left(\frac{m^2}{4\pi\mu^2}\right)+2 \gamma -3}{} \right],
\\
\delta a^{(0)}_{\Sigma}&=
-\frac{\lambda(6 m^2-\lambda ^2)}{96 \pi ^2 }
\Bigg[ \frac{2}{{(d-4)}}+2\gamma-\frac{8}{3}+ \log \left(\frac{m^2}{\pi\mu ^2}\right) \Bigg],
\\
\delta a^{(1)}_{\Sigma} &=
-\frac{\lambda }{16 \pi ^2 }\left[\frac{2}{(d-4)}+\gamma +\log \left(\frac{m^2}{4 \pi  \mu^2}\right)\right],
\\
\delta a^{(1)}_{M}&=
-\frac{m^2}{16 \pi ^2 } \left[\frac{2}{(d-4)}+\gamma -1+ {\log \left(\frac{m^2}{4 \pi  \mu^2}\right)}{}\right],
\\
\delta a^{(2)}_{M}&=
-\frac{1}{32 \pi ^2}\left[\frac{2}{ (d-4)}+\gamma  +\log \left(\frac{m^2}{4\pi\mu^2}\right)\right].
\end{split}
\end{align}
The renormalized parameters can be then defined as $a_{i,ren}^j:=a_{i,0}^j-\delta a_i^j$, 
where $a_{i,0}^j$ are the so-called bare parameters. 
Afterwards, one can write down a Callan--Symanzik equation involving the parameter $\mu$; 
the beta functions are as customarily defined 
as 
\begin{align}
\beta_i^j= \mu \frac{\partial a_{i,ren}^j}{\partial \mu}.
\end{align}
A direct computation, taking into account that $m$ and $\lambda$ are $\mu$-independent,
shows that in the present scheme the beta functions are\footnote{\old{Recall that we are not taking into account the quantum contributions of $\sigma$ at this point; 
the formulae in~\eqref{eq:beta_nonMS} should be understood as the contributions to the beta functions derived from the quantum fluctuations of $\phi$.}}
\begin{align}\label{eq:beta_nonMS}
 \begin{split}
\beta^{(0)}_{M} &=
\frac{m^4}{16 \pi ^2 } ,
\quad
\beta^{(0)}_{\Sigma}=
\frac{\lambda(6 m^2-\lambda ^2)}{48 \pi ^2 },
\quad
\beta^{(1)}_{\Sigma} =
\frac{\lambda }{8 \pi ^2 } ,
\\
\beta^{(1)}_{M}&=
\frac{m^2}{8 \pi ^2 } ,
\quad
\beta^{(2)}_{M}=
\frac{1}{16 \pi ^2}.  
 \end{split}
\end{align}
}

\subsection{{An alternative renormalization scheme}}
\old{In this subsection we will use a more physical scheme}, in spirit similar to the discussions in 
Refs.~\cite{Gorbar:2022xvb, Franchino-Vinas:2021lbl, Franchino-Vinas:2019upg, Franchino-Vinas:2018gzr, Donoghue:2015nba, Asorey:2003uf, Gorbar:2003yt, Gorbar:2002pw},
noting that 
the functions $C^{(1)}_{\Sigma}$ and $C^{(2)}_{M}$ play the role of running coupling constants. 
Indeed, suppose that we measure the bulk quartic coupling of $\sigma$ to have the value $c_2$ at a scale $q^2$.
Then its value at another scale $k^2$ will simply be given by $c_2+C^{(2)}_{M}(k^2,m)-C^{(2)}_{M}(q^2,m)$,
showing that our assertion is true (up to an experimentally determined constant). 
Additionally, the corresponding beta function reads
\begin{align}\label{eq:beta_m2}
 {\beta'}_M^{(2)}= \frac{\partial C^{(2)}_M(q^2,m)}{\partial \log q}
 =\frac{1}{32 \pi ^2}-\frac{m^2 \text{arccsch}\left(2 \sqrt{\frac{m^2}{q^2}}\right)}{8 \pi ^2 \sqrt{q^2 \left(4 m^2+q^2\right)}}.
\end{align}

\begin{figure}[h!]
\begin{center}
 \begin{minipage}{0.46\textwidth}
 \vspace{0.0cm}\includegraphics[width=.9\textwidth]{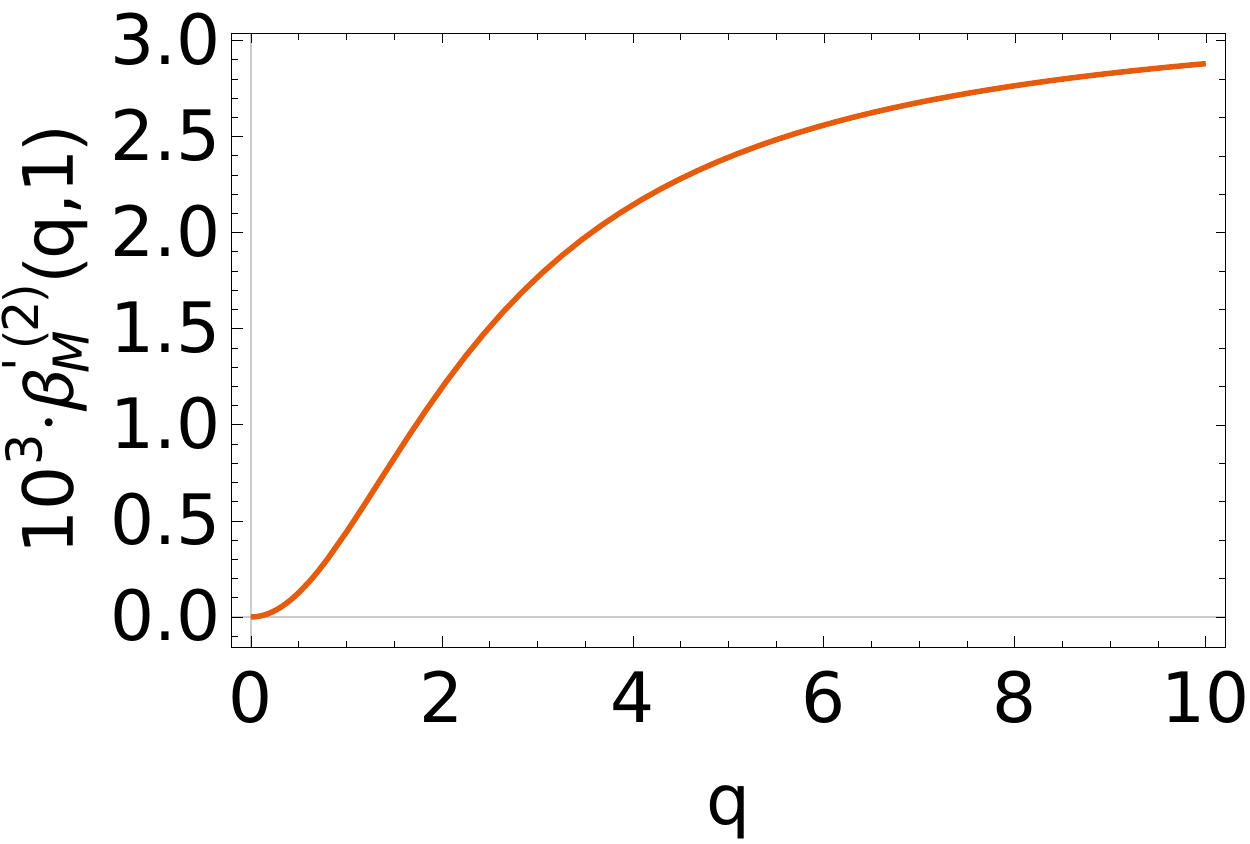} 
 \end{minipage}
\hspace{0.02\textwidth}
\begin{minipage}{0.46\textwidth}
 \vspace{0.0cm}\includegraphics[width=1.0\textwidth]{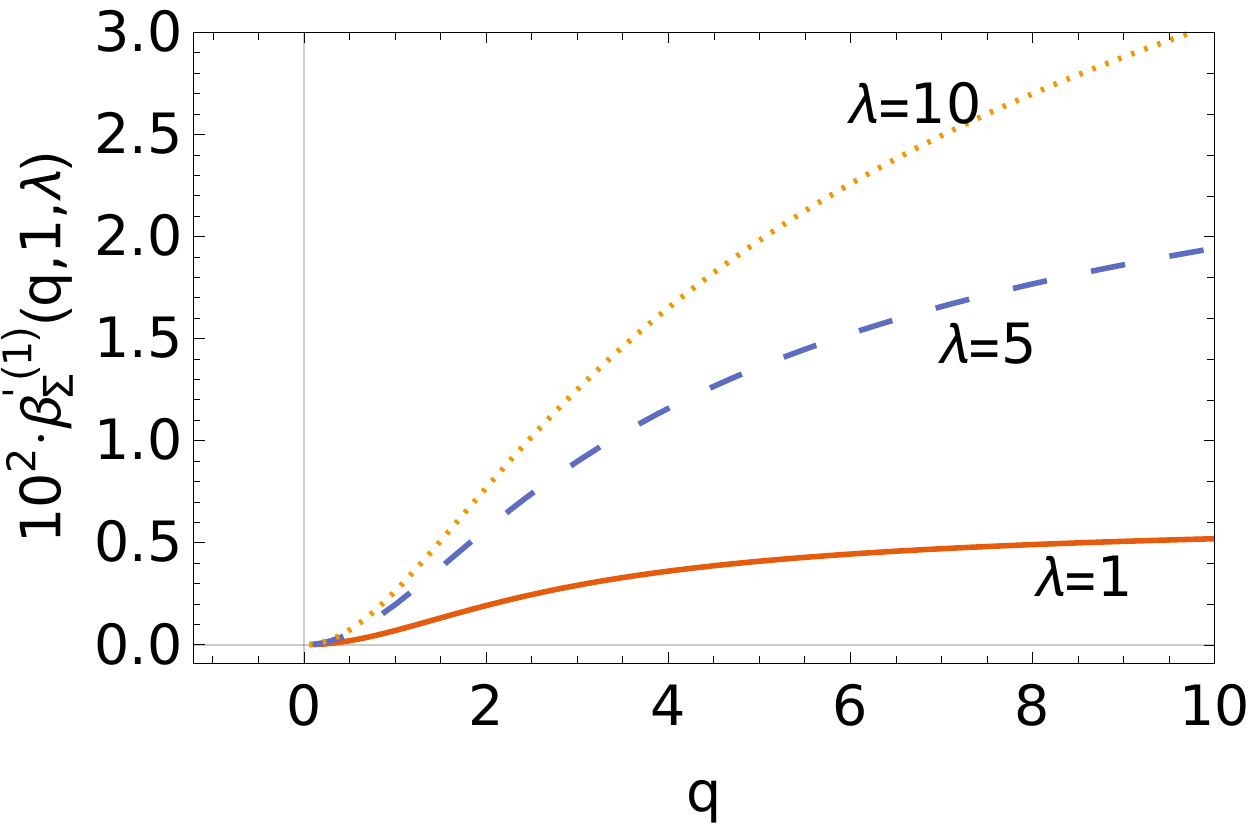} 
 \end{minipage}
 \end{center}
\caption{Rescaled beta functions $10^3\times {\beta'}^{(2)}_M(q,1)$ (left panel) and $10^2\times {\beta'}^{(1)}_\Sigma(q,1,\lambda)$ (right panel) as functions of the momentum scale $q$
(in arbitrary units). 
We have set the mass to unity; in the right panel, the different curves, from bottom to top, correspond to the values $\lambda=1$, 5 and 10 (in arbitrary units).
}
\label{fig:beta_functions}
 \end{figure}

The situation is slightly different for $C^{(1)}_{\Sigma}$,
inasmuch as the relevant scale for its running is determined by the momentum perpendicular to the plate. 
On the one side, this means that our renormalization breaks Lorentz invariance, 
which was anyway already broken by the plate configuration.
On the other side, this is not a consequence of having neglected in Eq.~\eqref{eq:ea_form_factors} terms involving 
partial derivatives with respect to the parallel directions: 
for fields vanishing at infinity, their contributions \old{to $C_{\Sigma}^{(1)}$} 
will be boundary terms that simply vanish\footnote{\old{In general, this is no longer true for terms with higher powers of $\sigma$, such as the contribution $F_\Sigma^{(2)}$ in Eq.~\eqref{eq:ea_form_factors}.}}.\old{
Analogously to the previous calculation, we can compute the beta function for $a_\Sigma^{(1)}$ in this scheme:
\begin{align}\label{eq:beta_s1}
 \begin{split}
{\beta'}_\Sigma^{(1)}&= \frac{\partial C^{(1)}_\Sigma(q,m,\lambda)}{\partial \log q}
 \\
 &=\frac{\lambda  q^2}{16 \pi ^2 \left(\lambda ^2+q^2\right)}-\frac{\lambda ^2 q^2 \left(-\lambda ^2+8 m^2+q^2\right)}{32 \pi  \sqrt{4 m^2+q^2} \left(\lambda ^2+q^2\right)^2}
 \\
 &\hspace{0.5cm}+\frac{\lambda  \sqrt{q^2} \left(\lambda ^2 \left(2 m^2+q^2\right)-2 m^2 q^2\right) }{8 \pi ^2 \sqrt{4 m^2+q^2} \left(\lambda ^2+q^2\right)^2} \operatorname{arctanh}\left(\sqrt{\frac{q^2}{4 m^2+q^2}}\right)
 \\
 &\hspace{1.5cm}-\frac{\lambda ^2 q^2  \sqrt{\left(\lambda ^2-4 m^2\right) }}{8 \pi ^2 \left(\lambda ^2+q^2\right)^2} \operatorname{arccoth}\left(\frac{\lambda }{\sqrt{\lambda ^2-4 m^2}}\right).
\end{split}
\end{align}
}

\old{It is important to notice that, as a consequence of discussion in  the \hidepdo{previous}{precedent} paragraphs,}we are lead to the conclusion that
$C^{(1)}_{\Sigma}$ and $C^{(2)}_{M}$ do have physical meaning,
differently to the case of $C^{(0)}_{\Sigma}$\old{(at least, as long as we do not include 
a dynamical gravity).}The quantum field $\phi$ acts as a mediator between the boundary condition and the background $\sigma$,
such that, if we try to confine $\phi$, then we are also automatically  enforcing $\sigma$ to satisfy a boundary condition.
The nature of the latter is encoded in the form factor $F_{\Sigma}^{(1)}$. 

As a last comment, it is interesting to study the asymptotic expansion of the running coupling constants for large and small $k$ (or masses), to wit:
\begin{align}\label{eq:C1_expansion}
&C^{(1)}_{\Sigma}(k,m,\lambda) =
\begin{cases}
\frac{\lambda  }{32 \pi ^2}\bigg[ \log \left(\frac{k^2}{m^2}\right)-2  {+\pi   \lambda (k^2)^{-\frac{1}{2}}- \left(\lambda ^2-2 m^2\right)\log \left(\frac{k^2}{m^2}\right)  k^{-2}}
\\
\hu+\left(2 \lambda  \sqrt{\lambda ^2-4 m^2} \operatorname{arccoth}\left(\frac{\lambda }{\sqrt{\lambda ^2-4 m^2}}\right)+2 m^2\right) {k^{-2}}
\\
\hu\hu+O\left({(k^2)^{-3/2}}\right) \bigg],\quad  k^2 \gg m^2,
\\
\frac{\lambda ^2}{128 \pi  m}
\bigg[ 
1+\frac{2 \left(k^2-\lambda ^2\right)}{3 \pi  \lambda  m}+\frac{\lambda ^2-k^2}{16 m^2}
\\
\hspace{2cm}-\frac{k^4+\lambda ^4-\lambda ^2 k^2}{15 \pi  \lambda  m^3}
+\mathcal{O}\left(m^{-4}\right)
\bigg],\quad k^2 \ll m^2,
\end{cases}
\end{align}
\begin{align}\label{eq:C2_expansion}
 C^{(2)}_M(k^2,m)
 =
 \begin{cases}
\frac{1}{64 \pi ^2} \bigg[\left(\log \left(\frac{k^2}{m^2}\right)-2\right)+2 m^2 \left(\log \left(\frac{k^2}{m^2}\right)+1\right) {k^{-2}}
\\
\hu+{m^4 \left(-2 \log \left(\frac{k^2}{m^2}\right)+1\right)}{k^{-4}}+\mathcal{O}\left({k^{-6}}\right)
\bigg]
,\quad k^2\gg m^2,
\\
\frac{k^2}{384 \pi ^2 m^2}\left[1-\frac{k^2}{10 m^2}+\frac{k^4}{70 m^4}-\frac{17 k^6}{2240 m^6}+\mathcal{O}\left(m^{-8}\right)\right],\quad  k^2\ll m^2
 \end{cases}
\end{align}
For large masses,  they display the decoupling of the quantum field, analogue to the Appelquist and Carazzone result for QED~\cite{Appelquist:1974tg}
and generalizing the result for a Yukawa coupling without boundaries~\cite{Ferreiro:2021pec}.
Additionally, the coefficients in the expansion are local. Instead, for large $k^2$ we see the nonlocal character of the runnings.
In particular, we observe that both couplings diverge for $k\to \infty$. 
This implies that for situations involving high energy processes in the $d$th direction,
at the one-loop level,
the $\sigma$ field will have to obey a strong boundary condition on $\Sigma$, i.e. almost Dirichlet.\old{
One can see that the plots in Fig.~\ref{fig:beta_functions}, which display the behaviours of the beta functions with the momentum scale $q$ and the coupling $\lambda$, 
are in agreement with the \hidepdo{preceding}{precedent} description.
}



\hidepdo{}{
\subsection{Formal aspects of the form factors}\label{sec:psido}
We will close this section with a short digression on a more mathematical description of the form factors.
The ultimate goal will be to evidence one fact that is usually left aside in the literature:
in general, form factors can be formally discussed in the frame of pseudo-differential operators.
To begin, let us recall the following definitions borrowed from Ref.~\cite{Hoermander:LPDOIII}.

\begin{definition}[Symbols]
Let $m \in \mathbb{R}$ and $n\in \mathbb{N}$. 
The class of symbols $S^{m}(\mathbb{R}^{n}\times \mathbb{R}^n)$ 
consists of functions\footnote{$C^{\infty}$ denotes the space of infinitely differentiable functions.} (symbols) $a(x,\xi)\in C^{\infty}(\mathbb{R}^n\times \mathbb{R}^n)$ 
such that, for all multi-indices $\alpha$ and $\beta$, a constant $C_{\alpha,\beta}$ exists for which
\begin{align}\label{eq:def_symbol}
 | \partial_\xi^\alpha \partial^\beta_x a(x,\xi) | \leq C_{\alpha,\beta} (1+|\xi|)^{m-|\alpha| }, \quad x,\xi\in \mathbb{R}^n.
\end{align}

\end{definition}

 \begin{definition}[Pseudo-differential operator]\label{def:psido}
 If $a\in S^m$ and\footnote{$\mathscr{S}$ is Schwartz's space of functions whose derivatives are rapidly decreasing.} $u\in \mathscr{S}$, then 
 \begin{align}\label{eq:def_psido}
  a(x, D) u(x) =(2\pi)^{-n} \int \dx[\xi]\,  e^{\mathi\, x \cdot \xi} \, a(x,\xi)\, \tilde u(\xi)
 \end{align}
defines a function $a(x,D)u\in \mathscr{S}$. One calls $a(x,D)$ a pseudo-differential operator of order $m$.
\end{definition}

\begin{proposition}
\old{Consider couplings $\lambda,m>0$.} Then the form factors $C_{\Sigma}^{(1)},C^{(2)}_{M}\in S^{j}$ for any $j>0$,
and therefore can be used to define pseudo-differential operators.  
\end{proposition}
\begin{proof}
Notice first that $C_{\Sigma}^{(1)}$ and $C^{(2)}_{M}$ have no $x$ dependence and $k$ plays the role of $\xi$ in Definition~\ref{def:psido}.
For the terms made only of algebraic functions of $k$,
a bound as in Eq.~\eqref{eq:def_symbol} can be easily obtained; moreover, these terms are infinitely differentiable,
since as assumption we have $m,\lambda>0$.

The terms that involve inverse hyperbolic functions display a twofold complication. 
First, they seem to have a singular behaviour for $k=0$. 
This is just apparent, what can be shown using a Taylor expansion of the relevant functions around $k=0$.
Second, the corresponding  large $k$ expansion yields $\log( k^2)$ terms, see Eqs.~\eqref{eq:C1_expansion} and \eqref{eq:C2_expansion} below. 
However, after a first differentiation one obtains algebraic functions of $k$ and a bound as in \eqref{eq:def_psido} can be straightforwardly obtained. 
The fact that they belong to $C^{\infty}$ is thus also proved.
\end{proof}

Consequently, to analyze the form factors we can employ all the machinery of the theory of pseudo-differential operators. 
In particular, it is instructive to analyze under what circumstances the integrals in Eq.~\eqref{eq:ea_form_factors} are convergent.
It is known that if the symbol belongs to the class $S^0$, then the associated pseudo-differential operator is an endomorphism in 
the space $L^2(\mathbb{R}^n)$ of Lebesgue square-integrable functions, 
see \cite[Theorem 18.1.11]{Hoermander:LPDOIII}.
Although in the present case the $\log k^2$ behaviour for large $k$ prevents us from using this result, 
we can envisage two alternatives.

The first one is to stick to $L^2(\mathbb{R}^n)$ as the domain on which the form factors act, 
so that the image of the latter will have to be interpreted in terms of generalized functions. 
This has been followed for example in Ref.~\cite{Franchino-Vinas:2021lbl}.

The remaining option is to restrict further the domain, 
in order to obtain a square-integrable function after applying the form factors~\cite{Mazzitelli:2011st}.
If we follow this possibility, we can subtract to the form factors a term $\log(1+k^2)$ with an appropriate coefficient,
obtaining thus a symbol in the class $S^0$. The singular contribution is then given by a  factor $\log(1+k^2)$.
As an example, in $n=1$ a sufficient, albeit not necessary set  of conditions on $\sigma^2$ that imply $C_{\Sigma}^{(1)}\sigma^2,C^{(2)}_{M}\sigma^2\in L^2(\mathbb{R})$ 
can be obtained using the lemma of Riemann--Lebesgue\footnote{$C^2(\mathbb{R})$ is the space of functions of a real variable with continuous derivatives up to second order. 
$L^1(\mathbb{R})$ is the space of absolutely Lebesgue integrable functions of a real variable.}:
if $\sigma^2\in C^{2}(\mathbb{R})$ and $\sigma^2,(\sigma^2){'},(\sigma^2){''}\in L^1(\mathbb{R})$, then $C_{\Sigma}^{(1)}\sigma^2,C^{(2)}_{M}\sigma^2\in L^2(\mathbb{R})$ as desired.
}


\section{Conclusions and outlook}
From the mathematical point of view, we have determined an infinite number of GSDW coefficients.
Bearing in mind the universality property discussed in Remarks~\ref{comm:order0},~\ref{comm:order1} and \ref{comm:order2},
these results may be an useful guide in more involved computations.\old{We 
have also discussed, in Sec.~\ref{sec:bcs}, the connection of these computations with the GSDW of the Dirichlet and Robin problems.}

As a next step it will be interesting to consider small curvature corrections to our problem.
One could think for example in curving the plates, i.e. generating an extrinsic curvature on $\Sigma$, 
or studying totally geodesic plates in a curved manifold.
These problems may be analyzed either by expanding in powers of the curvature contributions 
or considering curved configurations in which form factors may be obtained to all order in the curvatures.
In the latter case, homogeneous spaces are probably the most natural candidate. 

From the physical point of view we have analyzed, as far as we know for the first time in the literature, 
surface form factors of a Yukawa theory and their implications.
The most immediate consequence is the emergence of an energy-dependent semitransparent 
boundary condition satisfied by the background field.

One attractive possibility is to introduce a self-interacting term in the action~\eqref{eq:action}. 
If one adds a quartic interaction, i.e. $\lambda\,\phi^4/12$, 
then Eq.~\eqref{eq:qft_operator} acquires an additional term,
which reads $\lambda \varphi^2$. 
The field $\varphi$ is the classical field~\cite{Peskin:1995ev} and it is clear that 
its contribution to the effective action can be obtained by replacing $\sigma^2\to \sigma^2+\lambda \varphi^2$.
Doing so, one can readily compare the divergent terms with the counterterms found in Ref.~\cite{Bordag:2004rx}.
By the discussion in the present manuscript, the semitransparent boundary condition satisfied by the classical field $\varphi$ 
acquires then a dependence with the energy involved in a given physical process.
However, a subtle point is that the classical field is expected to be discontinuous;
\hidepdo{in such case, 
the action of the form factors on it will have to be interpreted in terms of generalized functions
and a detailed analysis will be required.}{in such case, taking into account the discussion in Sec.~\ref{sec:psido}, 
the action of the form factors on it will have to be interpreted in terms of generalized functions
and a detailed analysis will be required.}

Natural generalizations include the study of more general physical models and boundary conditions.
These lines are currently being explored.


\section*{Acknowledgements}
S.A.F. is indebted to D.F.~Mazzitelli and D. Vassilevich
for interesting discussions; 
he is also grateful to J.~Edwards, S.~Evans and G. Schaller for their suggestions.
S.A.F. acknowledges support from Helmholtz-Zentrum Dresden-Rossendorf (HZDR), PIP 11220200101426CO Consejo
Nacional de Investigaciones Cient\'ificas y T\'ecnicas (CO\-NI\-CET) and Project 11/X748, UNLP.


\appendix
\section{Form factors in arbitrary dimensions}\label{app:form_factors_d}
The explicit expressions for the form factors appearing in the effective action of Eq.~\eqref{eq:ea_form_factors}, 
for arbitrary Euclidean spacetime dimensions $d$, read
\begin{align}
 F^{(1)}_{M}(m):&=2^{-d-1} \pi ^{-\frac{d}{2}} m^{d-2}  \Gamma \left(1-\frac{d}{2}\right),
 \\
 \begin{split}
F^{(2)}_{M}\big(k^2,m\big) :&=-2^{2-2 d} \pi ^{-\frac{d}{2}}  \Gamma \left(2-\frac{d}{2}\right) \left(k^2+4 m^2\right)^{\frac{d}{2}-2} 
 \\
 &\hspace{3cm}\times\, _2F_1\left(\frac{1}{2},2-\frac{d}{2};\frac{3}{2};\frac{k^2}{k^2+4 m^2}\right),
  \end{split}
\end{align}
\begin{align}
 \begin{split}
F^{(1)}_{\Sigma}(k,m,\lambda):&=
 -\frac{2^{2-2 d} \pi ^{\frac{(-d+1)}{2}} \lambda  }{\left(k^2+\lambda ^2\right)}\Bigg\lbrace
 \lambda  \Gamma \left(\frac{3}{2}-\frac{d}{2}\right) \left(k^2+4 m^2\right)^{\frac{d}{2}-\frac{3}{2}}
 \\
 &\hspace{-0.5cm}+\frac{2 k^2 \Gamma \left(2-\frac{d}{2}\right) \left(k^2+4 m^2\right)^{\frac{d}{2}-2} \, _2F_1\left(\frac{1}{2},2-\frac{d}{2};\frac{3}{2};\frac{k^2}{k^2+4 m^2}\right)}{\sqrt{\pi }}
 \\
 &\hspace{1.0cm}+\frac{2 \lambda ^{d-2} \Gamma \left(2-\frac{d}{2}\right) \, _2F_1\left(\frac{3-d}{2},2-\frac{d}{2};\frac{5-d}{2};1-\frac{4 m^2}{\lambda ^2}\right)}{\sqrt{\pi } (d-3)}
 \Bigg\rbrace,
 \end{split}
\end{align}
\begin{align}
\begin{split}
 &\pi^{d/2}F^{(2)}_{\Sigma}\left(k_1,k_2,m,\lambda\right) 
 \\
 &=\Bigg[ \frac{2^{3-2 d}  \lambda   \Gamma \left(2-\frac{d}{2}\right) k_1^2\left(\lambda ^2 k_1+\lambda ^2 k_2+k_1^3-k_2 k_1^2\right) }{\left(k_1^2-k_2^2\right) k_2  \left(\lambda ^2+k_1^2\right){}^2}
 \\
 &\hspace{3cm}\times \left(k_1^2+4 m^2\right){}^{\frac{d}{2}-2} \, _2F_1\left(\frac{1}{2},2-\frac{d}{2};\frac{3}{2};\frac{k_1^2}{4 m^2+k_1^2}\right)
 \\
 &\hu+\frac{4^{1-d} \pi ^{\frac{1}{2}} \lambda ^2 \Gamma \left(\frac{3}{2}-\frac{d}{2}\right) k_1 \left(\lambda ^2+k_1^2-2 k_1 k_2\right) \left(k_1^2+4 m^2\right){}^{\frac{d-3}{2}}}{\left(k_1^2-k^2_2\right) k_2  \left(\lambda ^2+k_1^2\right){}^2}+\{k_1\leftrightarrow k_2\}\Bigg]
 \\
 &-\frac{2^{3-2 d}  \lambda   \Gamma \left(2-\frac{d}{2}\right) \left(k_1+k_2\right){}^2\left[(k_1+k_2)^2+4 m^2\right]{}^{\frac{d}{2}-2} }{k_1 k_2 \left[\lambda ^2+(k_1+k_2)^2 \right]}
 \\
 &\hspace{3cm}\times \, _2F_1\left(\frac{1}{2},2-\frac{d}{2};\frac{3}{2};\frac{\left(k_1+k_2\right){}^2}{4 m^2+(k_1+k_2)^2}\right)
 \\
 &+\frac{2^{-d-1}  \lambda ^3 \left(m^2\right)^{d/2} \Gamma \left(2-\frac{d}{2}\right)}{m^4 \left(\lambda ^2+k_1^2\right) \left(\lambda ^2+k_2^2\right)}
 -\frac{2^{-d-1} \pi ^{\frac{1}{2}-\frac{d}{2}} \lambda ^2 \Gamma \left(\frac{3}{2}-\frac{d}{2}\right) \left(\frac{1}{4} \left(k_1+k_2\right){}^2+m^2\right){}^{\frac{d-3}{2}}}{k_1 k_2 \left[\lambda ^2+(k_1^2+k_2^2)\right]}
 \\
 &-\frac{2^{4-2 d}   \Gamma \left(3-\frac{d}{2}\right) \lambda ^{d-1} }{(5-d) \left(\lambda ^2+k_1^2\right) \left(\lambda ^2+k_2^2\right)}\, _2F_1\left(\frac{4-d}{2}+1,\frac{5-d}{2};\frac{5-d}{2}+1;1-\frac{4 m^2}{\lambda ^2}\right)
 \\
 &+\frac{2^{3-2 d}  \Gamma \left(2-\frac{d}{2}\right)\left(5 \lambda ^4+k_1^2 \lambda ^2+k_2^2 \lambda ^2-4 k_1 k_2 \lambda ^2-2 k_1 k_2^3-3 k_1^2 k_2^2-2 k_1^3 k_2\right)}{(3-d) \left(\lambda ^2+k_1^2\right){}^2 \left(\lambda ^2+k_2^2\right){}^2 \left(\lambda ^2+k_1^2+k_2^2+2 k_1 k_2\right)}
 \\
 &\hspace{3cm}\times   \lambda ^{d-1} \,  k_1 k_2 \, {}_2F_1\left(\frac{2-d}{2}+1,\frac{3-d}{2};\frac{3-d}{2}+1;1-\frac{4 m^2}{\lambda ^2}\right).
\end{split}
\end{align}

\printbibliography

\end{document}